\documentclass[11pt]{article}%
\usepackage{amsfonts}
\usepackage{amsmath}
\usepackage{hyperref}
\usepackage{amssymb}
\usepackage{graphicx}%
\providecommand{\U}[1]{\protect\rule{.1in}{.1in}}
\newtheorem{theorem}{Theorem}

\newtheorem{corollary}[theorem]{Corollary}

\newtheorem{remark}[theorem]{Remark}

\numberwithin{equation}{section}
\numberwithin{theorem}{section}
\numberwithin{table}{section}
\newenvironment{proof}[1][Proof]{\noindent\textbf{#1.} }{\ \rule{0.5em}{0.5em}}
\begin{document}

\title{Semiparametric clustered overdispersed multinomial goodness-of-fit of
log-linear models\thanks{This paper was supported by the Spanish Grants
MTM2015-67057 and ECO2015-66593 from Ministerio de Econom\'{\i}a and
Competitividad.}}
\author{Alonso-Revenga, J. M.$^{1}$, Mart\'{\i}n, N.$^{2}$\thanks{Corresponding
author, E-mail: \href{mailto: nimartin@ucm.es}{nimartin@ucm.es}.}, Pardo,
L.$^{3}$\\$^{1}${\small Department of Statistics and O.R. III, Complutense University of
Madrid, Spain}\\$^{2}${\small Department of Statistics and O.R. II, Complutense University of
Madrid, Spain}\\$^{3}${\small Department of Statistics and O.R. I, Complutense University of
Madrid, Spain}}
\maketitle

\begin{abstract}
Traditionally, the Dirichlet-multinomial distribution has been recognized as a
key model for contingency tables generated by cluster sampling schemes. There
are, however, other possible distributions appropriate for these contingency
tables. This paper introduces new test-statistics capable to test log-linear
modeling hypotheses with no distributional specification, when the individuals
of the clusters are possibly homogeneously correlated. The estimator for the
intracluster correlation coefficient proposed in Alonso-Revenga et al. (2016),
valid for different cluster sizes, plays a crucial role in the construction of
the goodness-of-fit test-statistic.

\end{abstract}

\noindent\textbf{Keywords}{\small :} Clustered Multinomial Data; Consistent
Intracluster Correlation Estimator; Log-linear model; Overdispersion; Quasi
Minimum Divergence Estimator{\small .}

\section{Introduction\label{sec1}}

In studies of frequency data, often the observations are organized in
clusters. For clustered frequency data the classical statistical procedures
are not longer valid. For example, in a study of hospitalized pairs of
siblings, it is desired to study wether gender has any influence in
schizophrenic diagnosis. Since the two outcomes of every pair of siblings (a
cluster) are correlated for all the $N$ pairs of siblings, the assumption of
independence of all the $2N$ observations is violated and the classical
independence test of two categorical variables, gender and schizophrenic
diagnosis, is in principle useless. The same problem of invalidity of the
classical chi-square and likelihood ratio tests are presented with any
statistical model used for clustered frequencies.

Frequency data cross-classified according to $K$ variables, $(X_{1}%
,...,X_{K})$, having $X_{k}$ categories $1,2,...,I_{k}$, $k=1,...,K$, are the
so-called $K$-way contingency tables with $M=I_{1}\times I_{2}\times
\cdots\times I_{K}$ cells. In order to clarify the concepts and notation we
will focus our interest only on $K=2$ variables, $(X_{1},X_{2})$, with $I$ and
$J$ categories respectively, i.e. it has $M=I\times J$ cells denoted by pairs
lexicographically ordered as%
\[
\Omega=\{(1,1),(1,2),...,(1,J),....,(I,1),(I,2),...,(I,J)\},
\]
but it is possible to extend easily the same idea to $K$ variables. The
bidimensional random variable associated with the $\ell$-th cluster of size
$n_{\ell}$, $\ell=1,...,N$, being $N$ the number of clusters, is denoted as%
\[
(X_{1,h}^{(\ell)},X_{2,h}^{(\ell)}),\quad\ell=1,...,N,\;h=1,...,n_{\ell}.
\]
Let%
\[
\mathrm{I}_{S}(X_{1},X_{2})=\left\{
\begin{array}
[c]{ll}%
1, & \text{if }(X_{1},X_{2})\in S\\
0, & \text{if }(X_{1},X_{2})\notin S
\end{array}
\right.
\]
denote an indicator function of $S\subset\Omega$. Taking into account the
total count associated with cell $(i,j)$ is%
\begin{equation}
Y_{ij}^{(\ell)}=\sum_{h=1}^{n_{\ell}}\mathrm{I}_{\{(i,j)\}}(X_{1,h}^{(\ell
)},X_{2,h}^{(\ell)}),\quad\ell=1,...,N, \label{eq0a}%
\end{equation}
the $\ell$-th two-way frequency table in vector notation is given%
\[
\boldsymbol{Y}^{(\ell)}=(Y_{11}^{(\ell)},...,Y_{1J}^{(\ell)},...,Y_{I1}%
^{(\ell)},...,Y_{IJ}^{(\ell)})^{T},\quad\ell=1,...,N,
\]
where \textquotedblleft$^{T}$\textquotedblright\ denotes the transpose of a
vector or matrix. In what follows, it is assumed an homogeneous probability
for each individual felt in cell $(i,j)$ of the $\ell$-th cluster%
\[
p_{ij}(\boldsymbol{\theta})=\Pr(X_{1,h}^{(\ell)}=i,X_{2,h}^{(\ell)}%
=j),\quad\ell=1,...,N,\;h=1,...,n_{\ell},
\]
whose expression depends on an unknown $M_{0}$-dimensional parameter vector%
\[
\boldsymbol{\theta}=(\theta_{1},...,\theta_{M_{0}})^{T}\in%
\mathbb{R}
^{M_{0}},
\]
in terms of a log-linear model%
\begin{equation}
\boldsymbol{p}(\boldsymbol{\theta})=\frac{\exp\{\boldsymbol{W\theta}%
\}}{\boldsymbol{1}_{M}^{T}\exp\{\boldsymbol{W\theta}\}}, \label{eq0c}%
\end{equation}
where $M_{0}<M-1$,
\begin{equation}
\boldsymbol{p}(\boldsymbol{\theta})=(p_{11}(\boldsymbol{\theta}),...,p_{1J}%
(\boldsymbol{\theta}),...,p_{I1}(\boldsymbol{\theta}),...,p_{IJ}%
(\boldsymbol{\theta}))^{T} \label{eq0b}%
\end{equation}
and the design matrix, $\boldsymbol{W}$, is a full rank matrix, with column
vectors linearly independent with respect to the $M$-dimensional vector of
$1$'s, $\boldsymbol{1}_{M}=(1,...,1)^{T}$.

Under common correlation model for any pair of individuals $h$ and $s$
($h,s=1,...,n_{\ell},\;h\neq s$) of any cluster $\ell=1,...,N$, the
intracluster correlation coefficient is defined as%
\begin{align*}
\rho_{ij}^{2}  &  =\mathrm{Cor}[\mathrm{I}_{\{(i,j)\}}(X_{1,h}^{(\ell
)},X_{2,h}^{(\ell)}),\mathrm{I}_{\{(i,j)\}}(X_{1,s}^{(\ell)},X_{2,s}^{(\ell
)})]\\
&  =\frac{\mathrm{E}[\mathrm{I}_{\{(i,j)\}}(X_{1,h}^{(\ell)},X_{2,h}^{(\ell
)})\mathrm{I}_{\{(i,j)\}}(X_{1,s}^{(\ell)},X_{2,s}^{(\ell)})]-\mathrm{E}%
[\mathrm{I}_{\{(i,j)\}}(X_{1,h}^{(\ell)},X_{2,h}^{(\ell)})]E[\mathrm{I}%
_{\{(i,j)\}}(X_{1,s}^{(\ell)},X_{2,s}^{(\ell)})]}{\sqrt{\mathrm{Var}%
(\mathrm{I}_{\{(i,j)\}}(X_{1,h}^{(\ell)},X_{2,h}^{(\ell)}))\mathrm{Var}%
(\mathrm{I}_{\{(i,j)\}}(X_{1,s}^{(\ell)},X_{2,s}^{(\ell)}))}}\\
&  =\frac{\Pr(X_{1,h}^{(\ell)}=X_{1,s}^{(\ell)}=i,X_{2,h}^{(\ell)}%
=X_{2,s}^{(\ell)}=j)-p_{ij}^{2}(\boldsymbol{\theta})}{p_{ij}%
(\boldsymbol{\theta})\left(  1-p_{ij}(\boldsymbol{\theta})\right)  },\quad
\ell=1,...,N,\;h,s=1,...,n_{\ell},\;h\neq s
\end{align*}
(see Eldridge et al. (2009), for more details). In correlated clustered
overdispersed multinomial frequency data, in case of having homogeneous
intracluster correlation cell by cell, $\rho^{2}=\rho_{ij}^{2}$, $i=1,...I$,
$j=1,...,J$ and for this case, taking into account (\ref{eq0a}) and%
\begin{align*}
\mathrm{E}[\mathrm{I}_{\{(i,j)\}}(X_{1,h}^{(\ell)},X_{2,h}^{(\ell)})]  &
=p_{ij}(\boldsymbol{\theta}),\\
\mathrm{Cov}[\mathrm{I}_{\{(i,j)\}}(X_{1,h}^{(\ell)},X_{2,h}^{(\ell
)}),I_{\{(i,j)\}}(X_{1,s}^{(\ell)},X_{2,s}^{(\ell)})]  &  =\left\{
\begin{array}
[c]{ll}%
\rho^{2}p_{ij}(\boldsymbol{\theta})\left(  1-p_{ij}(\boldsymbol{\theta
})\right)  , & h\neq s\\
p_{ij}(\boldsymbol{\theta})\left(  1-p_{ij}(\boldsymbol{\theta})\right)  , &
h=s
\end{array}
\right.  ,
\end{align*}
it is proven that%
\begin{equation}
\mathrm{E}[\boldsymbol{Y}^{(\ell)}]=n_{\ell}\boldsymbol{p}(\boldsymbol{\theta
})\quad\text{and}\quad\mathrm{Var}[\boldsymbol{Y}^{(\ell)}]=\vartheta
_{n_{\ell}}n_{\ell}\boldsymbol{\Sigma}_{\boldsymbol{p}(\boldsymbol{\theta})},
\label{eq1}%
\end{equation}
where%
\begin{equation}
\vartheta_{n_{\ell}}=1+(n_{\ell}-1)\rho^{2}, \label{eq2}%
\end{equation}
is referred to as \textquotedblleft design effect\textquotedblright%
\ associated with the $\ell$-th cluster,%
\begin{equation}
\boldsymbol{\Sigma}_{\boldsymbol{p}(\boldsymbol{\theta})}=\boldsymbol{D}%
_{\boldsymbol{p}(\boldsymbol{\theta})}-\boldsymbol{p}(\boldsymbol{\theta
})\boldsymbol{p}^{T}(\boldsymbol{\theta}), \label{eq3}%
\end{equation}
and $\boldsymbol{D}_{\boldsymbol{p}(\boldsymbol{\theta})}$\ is the diagonal
matrix of $\boldsymbol{p}(\boldsymbol{\theta})$. Since $\mathrm{Var}%
[Y_{ij}^{(\ell)}]>0$, it holds $\vartheta_{n_{\ell}}=1+(n_{\ell}-1)\rho^{2}>0$
for $\ell=1,...,N$\ and thus $\rho^{2}>-1/(\max\{n_{\ell}\}_{\ell=1}^{N}-1)$,
but in practice it is assumed that $\rho^{2}\geq0$. This is just the reason
why these models are termed \textquotedblleft overdispersed
models\textquotedblright. In particular, for $\rho^{2}=0$ all the frequency
tables are multinomial.

Correlated clustered multinomial frequency data have been dealt in the
statistical literature since many years ago through two different approaches.
Following Choi and McHugh (1989), the design-based approach provides
inferences with respect to the sampling distribution of estimates over
repetitions of the same design. The works of Fellegi (1980), Holt et al.
(1980), Rao and Scott (1981,1984), Bedrick (1983), Landis et al (1984), Koch
et al. (1975), Fay (1985), as well as references therein are good examples of
this approach. On the other hand, Altham (1976), Cohen (1976), Brier (1980),
Fienberg (1979), Men\'{e}ndez et al. (1995, 1996) postulate a probability
distribution to model the sample data. Dirichlet-multinomial is, historically,
the first suitable distribution to modelize homogeneously correlated clustered
overdispersed multinomial frequency with a fixed cluster size (see Mosimann,
1962). Later, Cohen (1976) and Altham (1976) proposed the $n$-inflated
distribution and more recently, Morel and Nagaraj (1993) proposed the
random-clumped distribution. The zero-inflated binomial distribution falls
also inside this family of homogeneously correlated clustered overdispersed
multinomial frequency data. Details about these distributions can be found in
Alonso-Revenga et al. (2016). In the current paper and in Alonso-Revenga et
al. (2016) a third approach is presented, different from the previous ones,
based on the sole knowledge of the vector mean and the variance-covariance
matrix of the distribution, given in (\ref{eq1}), associated with the
generator of the sample data. For log-linear modeling no distribution
assumption is required if the quasi minimum $\phi$-divergence estimators are
used. In the following we shall assume that the data are generated by a
population verifying (\ref{eq1}). One of the strengths of this methodology, is
that the proposed consistent estimator for $\rho^{2}$ is semi-parametric and
it exhibits by far a better behavior with regard to the mean square error
(MSE) in comparison with the existing estimation method, which is fully
non-parametric. This kind of estimators are specially appealing for improving
the behavior of the existing goodness-of-fit tests for log-linear models, with
regard to the exact sizes and powers. The second strength of this methodology,
is the flexibility in being applicable for different cluster sizes.

For the \emph{semiparametric clustered overdispersed multinomial
goodness-of-fit of log-linear models}, the interest lays on testing wether it
holds a particular log-linear model%
\begin{equation}
H_{0}:\;\boldsymbol{p}(\boldsymbol{\theta})=\frac{\exp\{\boldsymbol{W\theta
}\}}{\boldsymbol{1}_{M}^{T}\exp\{\boldsymbol{W\theta}\}}\quad\text{vs.}\quad
H_{1}:\;\boldsymbol{p}(\boldsymbol{\theta})\neq\frac{\exp\{\boldsymbol{W\theta
}\}}{\boldsymbol{1}_{M}^{T}\exp\{\boldsymbol{W\theta}\}}. \label{eq5}%
\end{equation}

\section{Asymptotic Goodness-Of-Fit (GOF) test-statistics for equal cluster
sizes\label{sec2}}

For the frequency tables and the probability vectors, a single index notation
is preferred, since it covers any value, $K$, for the dimension of the
contingency table. This means that the probability vector
\[
\boldsymbol{p}(\boldsymbol{\theta})=(p_{1}(\boldsymbol{\theta}),...,p_{M}%
(\boldsymbol{\theta}))^{T},
\]
and the $\ell$-th frequency table%
\begin{equation}
\boldsymbol{Y}^{(\ell)}=(Y_{1}^{(\ell)},...,Y_{M}^{(\ell)})^{T},\quad
\ell=1,...,N, \label{distr}%
\end{equation}
are valid to represent double index elements ordered as (\ref{eq0b}) when
$K=2$ ($M=IJ$), as well as to generalize for any value of $K$ when the
$K$-tuples are lexicographically ordered ($M=\prod_{k=1}^{K}I_{k}$). The
$M$-dimensional vector obtained from collapsing the whole data,
$\boldsymbol{Y}^{(\ell)}$, $\ell=1,...,N$, is denoted by%
\[
\boldsymbol{Y}=\sum_{\ell=1}^{N}\boldsymbol{Y}^{(\ell)}%
\]
and $MN$-dimensional vector which gathers the whole data, $\boldsymbol{Y}%
^{(\ell)}$, $\ell=1,...,N$, by%
\[
\widetilde{\boldsymbol{Y}}=(\boldsymbol{Y}^{(1)T},...,\boldsymbol{Y}%
^{(N)T})^{T}.
\]

In this section, a family of GOF test-statistics for testing (\ref{eq5}) with
equal cluster sizes is introduced. In the following section the case of
unequal cluster sizes is treated. Some preliminary results related to the
estimators of the probability vector, derived in Alonso-Revenga et al. (2016),
are first introduced. The non-parametric estimator of $\boldsymbol{p}%
(\boldsymbol{\theta})$, based on $N$ clusters of sizes $n_{\ell}=n$,
$\ell=1,...,N$, is the $M$-dimensional vector of relative frequencies obtained
collapsing the $N$ frequency tables $\boldsymbol{Y}^{(\ell)}$, $\ell=1,...,N$,%
\[
\widehat{\boldsymbol{p}}=\frac{1}{nN}\boldsymbol{Y}=\frac{1}{nN}\sum_{\ell
=1}^{N}\boldsymbol{Y}^{(\ell)}=\frac{1}{N}\sum_{\ell=1}^{N}%
\widehat{\boldsymbol{p}}^{(\ell)},
\]
where $\widehat{\boldsymbol{p}}^{(\ell)}=\tfrac{1}{n}\boldsymbol{Y}^{(\ell)}%
$\ represents the non-parametric estimator of $\boldsymbol{p}%
(\boldsymbol{\theta})$\ based exclusively on the $\ell$-th cluster.

Based on the collapsed table, $\boldsymbol{Y}$, the quasi minimum $\phi
$-divergence estimator (QM$\phi$E) of $\boldsymbol{\theta}$\ in (\ref{eq0c})
is defined as%
\[
\widehat{\boldsymbol{\theta}}_{\phi}=\widehat{\boldsymbol{\theta}}_{\phi
}\left(  \boldsymbol{Y}\right)  =\arg\min_{\theta\in\Theta}d_{\phi
}(\widehat{\boldsymbol{p}},\boldsymbol{p}\left(  \boldsymbol{\theta}\right)
),
\]
where $\phi\left(  x\right)  $ is a convex function, $x>0$, such that at
$x=1$,\ $\phi\left(  1\right)  =0$,\ $\phi^{\prime}\left(  1\right)  =0$,
$\phi^{\prime\prime}\left(  1\right)  >0$, at $x=0$,\ $0\phi\left(
0/0\right)  =0$, $0\phi\left(  p/0\right)  =\underset{u\rightarrow\infty
}{\lim}p\phi\left(  u\right)  /u$, and%
\begin{equation}
\mathrm{d}_{\phi}\left(  \widehat{\boldsymbol{p}},\boldsymbol{p}\left(
\boldsymbol{\theta}\right)  \right)  =\sum\limits_{r=1}^{M}p_{r}\left(
\boldsymbol{\theta}\right)  \phi\!\left(  \frac{\widehat{p}_{r}}{p_{r}\left(
\boldsymbol{\theta}\right)  }\right)  \label{eq6}%
\end{equation}
is the $\phi$-divergence between the probability vectors
$\widehat{\boldsymbol{p}}$ and $\boldsymbol{p}\left(  \boldsymbol{\theta
}\right)  $. For more details about $\phi$-divergence\ measures see Cressie
and Pardo (2002) and Pardo (2006).

The \emph{quasi-maximum likelihood estimator} (QMLE) of $\boldsymbol{\theta}$,
denoted by $\widehat{\boldsymbol{\theta}}$, is a particular case of the
QM$\phi$E by replacing the $\phi$-divergence by the Kullback divergence
between the probability vectors $\widehat{\boldsymbol{p}}$ and $\boldsymbol{p}%
\left(  \boldsymbol{\theta}\right)  $, i.e.,
\begin{align*}
&  \widehat{\boldsymbol{\theta}}=\widehat{\boldsymbol{\theta}}\left(
\boldsymbol{Y}\right)  =\arg\min_{\theta\in\Theta}\mathrm{d}_{Kullback}%
(\widehat{\boldsymbol{p}},\boldsymbol{p}\left(  \boldsymbol{\theta}\right)
),\\
&  \mathrm{d}_{Kullback}(\widehat{\boldsymbol{p}},\boldsymbol{p}\left(
\boldsymbol{\theta}\right)  )=\sum\limits_{r=1}^{M}\widehat{p}_{r}\log
\frac{\widehat{p}_{r}}{p_{r}\left(  \boldsymbol{\theta}\right)  },
\end{align*}
or equivalently $\widehat{\boldsymbol{\theta}}=\widehat{\boldsymbol{\theta}%
}\left(  \boldsymbol{Y}\right)  =\arg\min_{\theta\in\Theta}d_{\phi
}(\widehat{\boldsymbol{p}},\boldsymbol{p}\left(  \boldsymbol{\theta}\right)
)$, with $\phi(x)=x\log x-x+1$. Since the QM$\phi$Es are invariant estimators,%
\[
\boldsymbol{p}(\widehat{\boldsymbol{\theta}}_{\phi})=\frac{\exp
\{\boldsymbol{W}\widehat{\boldsymbol{\theta}}_{\phi}\}}{\boldsymbol{1}_{M}%
^{T}\exp\{\boldsymbol{W}\widehat{\boldsymbol{\theta}}_{\phi}\}}%
\]
is the QM$\phi$Es\ of $\boldsymbol{p}\left(  \boldsymbol{\theta}\right)  $.

\begin{theorem}
\label{Th1}The asymptotic distribution of the difference between the
non-parametric estimator and the QM$\phi E$ of $\boldsymbol{p}%
(\boldsymbol{\theta})$, with $N$ clusters of size $n$, is%
\[
\sqrt{N}(\widehat{\boldsymbol{p}}-\boldsymbol{p}(\widehat{\boldsymbol{\theta}%
}_{\phi_{2}}))\overset{\mathcal{L}}{\underset{N\rightarrow\infty
}{\longrightarrow}}\mathcal{N}(\boldsymbol{0}_{M},\tfrac{\vartheta_{n}}%
{n}(\boldsymbol{\Sigma}_{\boldsymbol{p}(\boldsymbol{\theta}_{0})}%
-\boldsymbol{\Sigma}_{\boldsymbol{p}(\boldsymbol{\theta}_{0})}\boldsymbol{W}%
\left(  \boldsymbol{W}^{T}\boldsymbol{\Sigma}_{\boldsymbol{p}%
(\boldsymbol{\theta}_{0})}\boldsymbol{W}\right)  ^{-1}\boldsymbol{W}%
^{T}\boldsymbol{\Sigma}_{\boldsymbol{p}(\boldsymbol{\theta}_{0})})),
\]
where $\boldsymbol{\theta}_{0}$\ is the unknown true value of
$\boldsymbol{\theta}$.
\end{theorem}

\begin{proof}
By following (A.5) and (A.3) in the proof of Theorem 2.2 of Alonso-Revenga et
al. (2016, Section A.3), it holds%
\[
\sqrt{N}(\boldsymbol{p}(\widehat{\boldsymbol{\theta}}_{\phi})-\boldsymbol{p}%
(\boldsymbol{\theta}_{0}))=\boldsymbol{\Sigma\boldsymbol{_{\boldsymbol{p}%
\left(  \theta_{0}\right)  }}W}\sqrt{N}(\widehat{\boldsymbol{\theta}}_{\phi
}-\boldsymbol{\theta}_{0})+o_{p}\left(  \boldsymbol{1}_{M}\right)  ,
\]
and%
\[
\sqrt{N}(\widehat{\boldsymbol{\theta}}_{\phi}-\boldsymbol{\theta}%
_{0})=\boldsymbol{D}_{\boldsymbol{p}\left(  \theta_{0}\right)  }^{-1/2}\left(
\boldsymbol{A}^{T}(\boldsymbol{\theta}_{0})\boldsymbol{A}(\boldsymbol{\theta
}_{0})\right)  ^{-1}\boldsymbol{A}^{T}(\boldsymbol{\theta}_{0})\sqrt{N}\left(
\widehat{\boldsymbol{p}}-\boldsymbol{p}\left(  \boldsymbol{\theta}_{0}\right)
\right)  +o_{p}\left(  \boldsymbol{1}_{M_{0}}\right)  ,
\]
where%
\[
\boldsymbol{A}(\boldsymbol{\theta}_{0})=\boldsymbol{D}_{\boldsymbol{p}\left(
\theta_{0}\right)  }^{-1/2}\boldsymbol{\Sigma\boldsymbol{_{\boldsymbol{p}%
\left(  \theta_{0}\right)  }}W}.
\]
Plugging $\sqrt{N}(\widehat{\boldsymbol{\theta}}_{\phi}-\boldsymbol{\theta
}_{0})$ into the expression of $\sqrt{N}(\boldsymbol{p}%
(\widehat{\boldsymbol{\theta}}_{\phi})-\boldsymbol{p}(\boldsymbol{\theta}%
_{0}))$\ we get%
\[
\sqrt{N}(\boldsymbol{p}(\widehat{\boldsymbol{\theta}}_{\phi})-\boldsymbol{p}%
(\boldsymbol{\theta}_{0}))=\boldsymbol{A}(\boldsymbol{\theta}_{0})\left(
\boldsymbol{A}^{T}(\boldsymbol{\theta}_{0})\boldsymbol{A}(\boldsymbol{\theta
}_{0})\right)  ^{-1}\boldsymbol{A}^{T}(\boldsymbol{\theta}_{0})\sqrt{N}\left(
\widehat{\boldsymbol{p}}-\boldsymbol{p}\left(  \boldsymbol{\theta}_{0}\right)
\right)  +o_{p}\left(  \boldsymbol{1}_{M}\right)  ,
\]
and subtracting the expressions on both sides of the equality to
$\widehat{\boldsymbol{p}}-\boldsymbol{p}\left(  \boldsymbol{\theta}%
_{0}\right)  $,
\[
\sqrt{N}(\widehat{\boldsymbol{p}}-\boldsymbol{p}(\widehat{\boldsymbol{\theta}%
}_{\phi}))=\left(  \boldsymbol{I}_{M}-\boldsymbol{A}(\boldsymbol{\theta}%
_{0})\left(  \boldsymbol{A}^{T}(\boldsymbol{\theta}_{0})\boldsymbol{A}%
(\boldsymbol{\theta}_{0})\right)  ^{-1}\boldsymbol{A}^{T}(\boldsymbol{\theta
}_{0})\right)  \sqrt{N}\left(  \widehat{\boldsymbol{p}}-\boldsymbol{p}\left(
\boldsymbol{\theta}_{0}\right)  \right)  +o_{p}\left(  \boldsymbol{1}%
_{M}\right)  .
\]
On the other hand, by applying the Central Limit Theorem%
\begin{equation}
\sqrt{N}\left(  \widehat{\boldsymbol{p}}-\boldsymbol{p}\left(
\boldsymbol{\theta}_{0}\right)  \right)  \overset{\mathcal{L}%
}{\underset{N\rightarrow\infty}{\longrightarrow}}\mathcal{N(}\boldsymbol{0}%
_{M},\tfrac{\vartheta_{n}}{n}\boldsymbol{\Sigma}_{\boldsymbol{p}\left(
\boldsymbol{\theta}_{0}\right)  }) \label{eq14}%
\end{equation}
(see Alonso-Revenga et al. (2016), eq. (3.1)), from which the asymptotic
distribution of $\sqrt{N}(\widehat{\boldsymbol{p}}-\boldsymbol{p}%
(\widehat{\boldsymbol{\theta}}_{\phi}))$ is an $M$-dimensional central normal
with variance-covariance matrix equal to%
\begin{align*}
&  \tfrac{\vartheta_{n}}{n}\left(  \boldsymbol{I}_{M}-\boldsymbol{A}%
(\boldsymbol{\theta}_{0})\left(  \boldsymbol{A}^{T}(\boldsymbol{\theta}%
_{0})\boldsymbol{A}(\boldsymbol{\theta}_{0})\right)  ^{-1}\boldsymbol{A}%
^{T}(\boldsymbol{\theta}_{0})\right)  \boldsymbol{\Sigma}_{\boldsymbol{p}%
\left(  \boldsymbol{\theta}_{0}\right)  }\left(  \boldsymbol{I}_{M}%
-\boldsymbol{A}(\boldsymbol{\theta}_{0})\left(  \boldsymbol{A}^{T}%
(\boldsymbol{\theta}_{0})\boldsymbol{A}(\boldsymbol{\theta}_{0})\right)
^{-1}\boldsymbol{A}^{T}(\boldsymbol{\theta}_{0})\right) \\
&  =\tfrac{\vartheta_{n}}{n}(\boldsymbol{\Sigma}_{\boldsymbol{p}%
(\boldsymbol{\theta}_{0})}-\boldsymbol{\Sigma}_{\boldsymbol{p}%
(\boldsymbol{\theta}_{0})}\boldsymbol{W}\left(  \boldsymbol{W}^{T}%
\boldsymbol{\Sigma}_{\boldsymbol{p}(\boldsymbol{\theta}_{0})}\boldsymbol{W}%
\right)  ^{-1}\boldsymbol{W}^{T}\boldsymbol{\Sigma}_{\boldsymbol{p}%
(\boldsymbol{\theta}_{0})}).
\end{align*}
The last equality comes from $\boldsymbol{\Sigma}_{\boldsymbol{p}%
(\boldsymbol{\theta}_{0})}\boldsymbol{D}_{\boldsymbol{p}\left(  \theta
_{0}\right)  }^{-1}\boldsymbol{\Sigma}_{\boldsymbol{p}(\boldsymbol{\theta}%
_{0})}=\boldsymbol{\Sigma}_{\boldsymbol{p}(\boldsymbol{\theta}_{0})}$ and
$\boldsymbol{D}_{\boldsymbol{p}\left(  \theta_{0}\right)  }^{-1}%
\boldsymbol{\Sigma}_{\boldsymbol{p}(\boldsymbol{\theta}_{0})}\boldsymbol{A}%
(\boldsymbol{\theta}_{0})=\boldsymbol{A}(\boldsymbol{\theta}_{0})$.\medskip
\end{proof}

The \emph{semi-parametric estimator} of $\vartheta_{n}$, via QM$\phi$Es, is%
\[
\widetilde{\vartheta}_{n,N,\phi}=\frac{X^{2}(\widetilde{\boldsymbol{Y}%
},\widehat{\boldsymbol{\theta}}_{\phi})}{(N-1)(M-1)},
\]
where%
\begin{equation}
X^{2}(\widetilde{\boldsymbol{Y}},\widehat{\boldsymbol{\theta}}_{\phi}%
)=\sum_{\ell=1}^{N}\left(  \boldsymbol{Y}^{(\ell)}-n\widehat{\boldsymbol{p}%
}\right)  ^{T}\tfrac{1}{n}\boldsymbol{D}_{\boldsymbol{p}%
(\widehat{\boldsymbol{\theta}}_{\phi})}^{-1}\left(  \boldsymbol{Y}^{(\ell
)}-n\widehat{\boldsymbol{p}}\right)  =n\sum_{\ell=1}^{N}\sum_{r=1}^{M}%
\frac{(\widehat{p}_{r}^{(\ell)}-\widehat{p}_{r})^{2}}{p_{r}%
(\widehat{\boldsymbol{\theta}}_{\phi})}. \label{eq17}%
\end{equation}
Similarly, the semi-parametric estimator of $\rho^{2}$, via QM$\phi$Es, is%
\[
\widetilde{\rho}_{n,N,\phi}^{2}=\frac{\widetilde{\vartheta}_{n,N,\phi}-1}%
{n-1}.
\]
Both, $\widetilde{\vartheta}_{n,N,\phi}$ and $\widetilde{\rho}_{n,N,\phi}^{2}%
$, are consistent estimators of $\vartheta$ and $\rho^{2}$\ respectively.

\begin{corollary}
\label{Cor2}The semiparametric clustered \emph{overdispersed chi-square GOF
test-statistic}, with $N$ clusters of size $n$, has the following asymptotic
distribution%
\[
\frac{X^{2}(\boldsymbol{Y},\widehat{\boldsymbol{\theta}}_{\phi})}%
{\widetilde{\vartheta}_{n,N,\phi}}\overset{\mathcal{L}}{\underset{N\rightarrow
\infty}{\longrightarrow}}\chi_{M-M_{0}-1}^{2},
\]
where%
\begin{equation}
X^{2}(\boldsymbol{Y},\widehat{\boldsymbol{\theta}}_{\phi}%
)=nN(\widehat{\boldsymbol{p}}-\boldsymbol{p}(\widehat{\boldsymbol{\theta}%
}_{\phi}))^{T}\boldsymbol{D}_{\boldsymbol{p}(\widehat{\boldsymbol{\theta}%
}_{\phi})}^{-1}(\widehat{\boldsymbol{p}}-\boldsymbol{p}%
(\widehat{\boldsymbol{\theta}}_{\phi}))=nN\sum_{r=1}^{M}\frac{(\widehat{p}%
_{r}-p_{r}(\widehat{\boldsymbol{\theta}}_{\phi}))^{2}}{p_{r}%
(\widehat{\boldsymbol{\theta}}_{\phi})}. \label{eq8}%
\end{equation}

\end{corollary}

\begin{proof}%
\begin{equation}
\frac{X^{2}(\boldsymbol{Y},\widehat{\boldsymbol{\theta}}_{\phi})}%
{\widetilde{\vartheta}_{n,N,\phi}}=\boldsymbol{Q}^{T}\boldsymbol{Q},
\label{eq15}%
\end{equation}
where%
\[
\boldsymbol{Q}=\sqrt{\frac{n}{\widetilde{\vartheta}_{n,N,\phi}}}%
\boldsymbol{D}_{\boldsymbol{p}(\widehat{\boldsymbol{\theta}}_{\phi})}%
^{-1/2}\sqrt{N}(\widehat{\boldsymbol{p}}-\boldsymbol{p}%
(\widehat{\boldsymbol{\theta}}_{\phi}))
\]
is an $M$-dimensional central normal with variance-covariance matrix equal to%
\begin{align*}
\boldsymbol{V}(\boldsymbol{\theta}_{0})  &  =\boldsymbol{D}_{\boldsymbol{p}%
(\boldsymbol{\theta}_{0})}^{-1}\boldsymbol{\Sigma}_{\boldsymbol{p}%
(\boldsymbol{\theta}_{0})}-\boldsymbol{D}_{\boldsymbol{p}%
(\widehat{\boldsymbol{\theta}}_{\phi})}^{-1/2}\boldsymbol{\Sigma
}_{\boldsymbol{p}(\boldsymbol{\theta}_{0})}\boldsymbol{W}\left(
\boldsymbol{W}^{T}\boldsymbol{\Sigma}_{\boldsymbol{p}(\boldsymbol{\theta}%
_{0})}\boldsymbol{W}\right)  ^{-1}\boldsymbol{W}^{T}\boldsymbol{\Sigma
}_{\boldsymbol{p}(\boldsymbol{\theta}_{0})}\boldsymbol{D}_{\boldsymbol{p}%
(\boldsymbol{\theta}_{0})}^{-1/2}\\
&  =\boldsymbol{D}_{\boldsymbol{p}(\boldsymbol{\theta}_{0})}^{-1}%
\boldsymbol{\Sigma}_{\boldsymbol{p}(\boldsymbol{\theta}_{0})}-\boldsymbol{A}%
(\boldsymbol{\theta}_{0})\left(  \boldsymbol{A}^{T}(\boldsymbol{\theta}%
_{0})\boldsymbol{A}(\boldsymbol{\theta}_{0})\right)  ^{-1}\boldsymbol{A}%
^{T}(\boldsymbol{\theta}_{0}),
\end{align*}
by applying the Slutsky's Theorem. The asymptotic distribution of a quadratic
form, such as (\ref{eq15}), with $\boldsymbol{V}(\boldsymbol{\theta}_{0})$
being idempotent, is a chi-square distribution with degrees of freedom equal
to the rank of $\boldsymbol{V}(\boldsymbol{\theta}_{0})$. The idempotence of
$\boldsymbol{V}(\boldsymbol{\theta}_{0})$\ is proven with similar arguments
given to obtain the variance-covariance matrix at the end of the proof of
Theorem \ref{Th1}. Finally, taking into account that for idempotent matrices
rank and trace are equivalent, and by properties of the trace of the product
of two matrices, it holds%
\begin{align*}
\mathrm{rank}(\boldsymbol{V}(\boldsymbol{\theta}_{0}))  &  =\mathrm{trace}%
(\boldsymbol{V}(\boldsymbol{\theta}_{0}))\\
&  =\mathrm{trace}(\boldsymbol{D}_{\boldsymbol{p}(\boldsymbol{\theta}_{0}%
)}^{-1}\boldsymbol{\Sigma}_{\boldsymbol{p}(\boldsymbol{\theta}_{0}%
)})-\mathrm{trace}(\boldsymbol{A}(\boldsymbol{\theta}_{0})\left(
\boldsymbol{A}^{T}(\boldsymbol{\theta}_{0})\boldsymbol{A}(\boldsymbol{\theta
}_{0})\right)  ^{-1}\boldsymbol{A}^{T}(\boldsymbol{\theta}_{0}))\\
&  =\mathrm{trace}(\boldsymbol{I}_{M}-\boldsymbol{p}(\boldsymbol{\theta}%
_{0}))-\mathrm{trace}(\boldsymbol{I}_{M_{0}})\\
&  =(M-1)-M_{0}.
\end{align*}

\end{proof}

\begin{remark}
The chi-square statistics given in (\ref{eq17}) and (\ref{eq8}) need some
clarifications, since under the same terminology arise totally different
ideas. While $X^{2}(\boldsymbol{Y},\widehat{\boldsymbol{\theta}}_{\phi})$ is
part of a GOF test-statistic, $X^{2}(\widetilde{\boldsymbol{Y}}%
,\widehat{\boldsymbol{\theta}}_{\phi})$ is part of an estimator constructed
through the trace of the quasi-variance-covariance matrix of
$\widetilde{\boldsymbol{Y}}$\ (see more details in Alonso-Revenga et al.
(2016)). Structurally, $X^{2}(\widetilde{\boldsymbol{Y}}%
,\widehat{\boldsymbol{\theta}}_{\phi})$ is quite different from the usual
chi-square test-statistics, since the total number of cells, $NM$, depends on
$N$, which increases to infinity.
\end{remark}

The $\phi$-divergence measures permit to construct either estimators as well
as test-statistics. Both of them do not need to be the same, for example in
the usual chi-square test-statistic with QMLEs, $X^{2}(\boldsymbol{Y}%
,\widehat{\boldsymbol{\theta}})=2nN\mathrm{d}_{\phi_{1}}%
(\widehat{\boldsymbol{p}},\boldsymbol{p}(\widehat{\boldsymbol{\theta}}%
_{\phi_{2}}))$, where $\phi_{1}(x)=\frac{1}{2}(x^{2}-1)$ and $\phi
_{2}(x)=x\log x-x+1$. In what is to follow, notation $\phi_{1}$ and $\phi_{2}%
$\ are used to distinguish the $\phi$ function of the $\phi$-divergences.

\begin{theorem}
\label{Th3}The semiparametric clustered \emph{overdispersed divergence based
GOF test-statistic}, with $N$ clusters of size $n$, has the following
asymptotic distribution%
\[
\frac{T^{\phi_{1}}(\boldsymbol{Y},\widehat{\boldsymbol{\theta}}_{\phi_{2}}%
)}{\widetilde{\vartheta}_{n,N,\phi_{2}}}\overset{\mathcal{L}%
}{\underset{N\rightarrow\infty}{\longrightarrow}}\chi_{M-M_{0}-1}^{2},
\]
where%
\begin{equation}
T^{\phi_{1}}(\boldsymbol{Y},\widehat{\boldsymbol{\theta}}_{\phi_{2}}%
)=\frac{2nN}{\phi_{1}^{\prime\prime}(1)}\mathrm{d}_{\phi_{1}}%
(\widehat{\boldsymbol{p}},\boldsymbol{p}(\widehat{\boldsymbol{\theta}}%
_{\phi_{2}})) \label{eq9}%
\end{equation}
and%
\[
\mathrm{d}_{\phi_{1}}(\widehat{\boldsymbol{p}},\boldsymbol{p}%
(\widehat{\boldsymbol{\theta}}_{\phi_{2}}))=\sum\limits_{r=1}^{M}%
p_{r}(\widehat{\boldsymbol{\theta}}_{\phi_{2}})\phi_{1}\!\left(
\frac{\widehat{p}_{r}}{p_{r}(\widehat{\boldsymbol{\theta}}_{\phi_{2}}%
)}\right)  .
\]

\end{theorem}

\begin{proof}
A second order Taylor expansion of $\mathrm{d}_{\phi_{1}}%
(\widehat{\boldsymbol{p}},\boldsymbol{p}(\widehat{\boldsymbol{\theta}}%
_{\phi_{2}}))$ around $\left(  \boldsymbol{p}\left(  \boldsymbol{\theta}%
_{0}\right)  ,\boldsymbol{p}\left(  \boldsymbol{\theta}_{0}\right)  \right)  $
needs derivatives of first and second order. Since
\[
\mathrm{d}_{\phi_{1}}\left(  \boldsymbol{p},\boldsymbol{q}\right)
=\sum\limits_{j=1}^{M}q_{j}\phi_{1}\left(  \frac{p_{j}}{q_{j}}\right)  ,
\]
and $\phi_{1}\left(  1\right)  =\phi_{1}^{\prime}\left(  1\right)  =0$, the
first order derivatives of the Taylor expansion are cancelled. The second
order derivatives yields%
\begin{align*}
\left.  \frac{\partial^{2}\mathrm{d}_{\phi_{1}}\left(  \boldsymbol{p}%
,\boldsymbol{p}\left(  \boldsymbol{\theta}_{0}\right)  \right)  }%
{\partial\boldsymbol{p}\partial\boldsymbol{p}^{T}}\right\vert
_{\boldsymbol{p=p}\left(  \boldsymbol{\theta}_{0}\right)  }  &  =\left.
\frac{\partial^{2}\mathrm{d}_{\phi_{1}}\left(  \boldsymbol{p}\left(
\boldsymbol{\theta}_{0}\right)  ,\boldsymbol{q}\right)  }{\partial
\boldsymbol{q}\partial\boldsymbol{q}^{T}}\right\vert _{\boldsymbol{q=p}\left(
\boldsymbol{\theta}_{0}\right)  }=-\left.  \frac{\partial^{2}\mathrm{d}%
_{\phi_{1}}\left(  \boldsymbol{p},\boldsymbol{q}\right)  }{\partial
\boldsymbol{q}\partial\boldsymbol{p}^{T}}\right\vert _{\boldsymbol{p=p}\left(
\boldsymbol{\theta}_{0}\right)  ,\boldsymbol{q=p}\left(  \boldsymbol{\theta
}_{0}\right)  }\\
&  =\phi_{1}^{\prime\prime}\left(  1\right)  \boldsymbol{D}_{\boldsymbol{p}%
\left(  \boldsymbol{\theta}_{0}\right)  }^{-1}.
\end{align*}
Hence,%
\[
\mathrm{d}_{\phi_{1}}(\widehat{\boldsymbol{p}},\boldsymbol{p}%
(\widehat{\boldsymbol{\theta}}_{\phi_{2}}))=\tfrac{1}{2}\phi_{1}^{\prime
\prime}\left(  1\right)  (\widehat{\boldsymbol{p}}-\boldsymbol{p}%
(\widehat{\boldsymbol{\theta}}_{\phi_{2}}))^{T}\boldsymbol{\boldsymbol{D}%
}_{\boldsymbol{p}\left(  \boldsymbol{\theta}_{0}\right)  }^{-1}%
(\widehat{\boldsymbol{p}}-\boldsymbol{p}(\widehat{\boldsymbol{\theta}}%
_{\phi_{2}}))+o(||\widehat{\boldsymbol{p}}-\boldsymbol{p}%
(\widehat{\boldsymbol{\theta}}_{\phi_{2}})||^{2}).
\]
Finally, since Theorem \ref{Th1} $o(N||\widehat{\boldsymbol{p}}-\boldsymbol{p}%
(\widehat{\boldsymbol{\theta}}_{\phi_{2}})||^{2})=o_{p}(1)$, and thus%
\[
\frac{T^{\phi_{1}}(\boldsymbol{Y},\widehat{\boldsymbol{\theta}}_{\phi_{2}}%
)}{\widetilde{\vartheta}_{n,N,\phi_{2}}}=\frac{1}{\widetilde{\vartheta
}_{n,N,\phi_{2}}}\frac{2Nn}{\phi_{1}^{\prime\prime}\left(  1\right)
\vartheta_{n}}\mathrm{d}_{\phi_{1}}(\widehat{\boldsymbol{p}},\boldsymbol{p}%
(\widehat{\boldsymbol{\theta}}_{\phi_{2}}))=\frac{1}{\widetilde{\vartheta
}_{n,N,\phi_{2}}}\boldsymbol{Q}^{T}\boldsymbol{Q}+o_{p}(1)=\frac
{X^{2}(\boldsymbol{Y},\widehat{\boldsymbol{\theta}}_{\phi_{2}})}%
{\widetilde{\vartheta}_{n,N,\phi_{2}}}+o_{p}(1),
\]
which means that $T^{\phi_{1}}(\boldsymbol{Y},\widehat{\boldsymbol{\theta}%
}_{\phi_{2}})/\widetilde{\vartheta}_{n,N,\phi_{2}}$ and $X^{2}(\boldsymbol{Y}%
,\widehat{\boldsymbol{\theta}}_{\phi_{2}})/\widetilde{\vartheta}_{n,N,\phi
_{2}}$ have the same asymptotic distribution, $\chi_{M-M_{0}-1}^{2}$,
according to Corollary \ref{Cor2}.\bigskip
\end{proof}

The following result is a particular case of Theorem \ref{Th3}, with $\phi
_{1}(x)=x\log x-x+1$.

\begin{corollary}
\label{Cor4}The semiparametric clustered \emph{overdispersed likelihood-ratio
GOF test-statistic}, with $N$ clusters of size $n$, has the following
asymptotic distribution%
\[
\frac{G^{2}(\boldsymbol{Y},\widehat{\boldsymbol{\theta}}_{\phi_{2}}%
)}{\widetilde{\vartheta}_{n,N,\phi_{2}}}\overset{\mathcal{L}%
}{\underset{N\rightarrow\infty}{\longrightarrow}}\chi_{M-M_{0}-1}^{2},
\]
where%
\begin{equation}
G^{2}(\boldsymbol{Y},\widehat{\boldsymbol{\theta}}_{\phi_{2}})=nN\sum
_{r=1}^{M}\widehat{p}_{r}\log\frac{\widehat{p}_{r}}{p_{r}%
(\widehat{\boldsymbol{\theta}}_{\phi_{2}})}. \label{eq10}%
\end{equation}

\end{corollary}

\section{Asymptotic Goodness-Of-Fit (GOF) test-statistics for unequal cluster
sizes\label{sec3}}

As introduction of this section a brief summary of the estimators given in
Alonso-Revenga et al. (2016) is presented. Organizing the frequency tables
associated with the clusters, according to their sizes, the double index in%
\[
\boldsymbol{Y}^{(g,\ell)}=(Y_{1}^{(g,\ell)},...,Y_{M}^{(g,\ell)})^{T}%
,\quad\ell=1,...,N,
\]
denotes the $g$-th frequency table of size $n_{\ell}$, $g=1,...,G$. In this
setting, the non-parametric estimator of $\boldsymbol{p}\left(
\boldsymbol{\theta}\right)  $\ is, according to Alonso-Revenga et al. (2016),
equal to
\[
\widehat{\boldsymbol{p}}=\frac{1}{\sum_{h=1}^{G}n_{h}N_{h}}\boldsymbol{Y}%
=\frac{1}{\sum_{h=1}^{G}n_{h}N_{h}}\sum_{g=1}^{G}\sum_{\ell=1}^{N}%
\boldsymbol{Y}^{(g,\ell)}=\sum_{g=1}^{G}w_{g}\widehat{\boldsymbol{p}}^{(g)},
\]
where%
\[
\boldsymbol{Y=}\sum_{g=1}^{G}\sum_{\ell=1}^{N}\boldsymbol{Y}^{(g,\ell)},
\]%
\begin{equation}
w_{g}=\frac{n_{g}N_{g}}{\sum_{h=1}^{G}n_{h}N_{h}}, \label{eq4}%
\end{equation}
and%
\begin{align*}
\widehat{\boldsymbol{p}}^{(g)}  &  =\frac{1}{n_{g}N}\sum_{\ell=1}^{N_{g}%
}\boldsymbol{Y}^{(g,\ell)}=\frac{1}{N}\sum_{\ell=1}^{N_{g}}%
\widehat{\boldsymbol{p}}^{(g,\ell)},\\
\widehat{\boldsymbol{p}}^{(g,\ell)}  &  =\frac{1}{n_{g}}\boldsymbol{Y}%
^{(g,\ell)}.
\end{align*}

Let
\begin{equation}
\vartheta_{n^{\ast}}=1+\rho^{2}\left(  n^{\ast}-1\right)  \in(1,n^{\ast}],
\label{eq7}%
\end{equation}
be design effect with a sample size equal to%
\begin{align*}
n^{\ast}  &  =\sum_{g=1}^{G}w_{g}^{\ast}n_{g},\\
w_{g}^{\ast}  &  =\frac{N_{g}^{\ast}n_{g}}{\sum\limits_{h=1}^{G}N_{h}^{\ast
}n_{h}}>0,\quad g=1,...,G,
\end{align*}
and $N_{g}^{\ast}\in(0,1]$\ an unknown value such that%
\[
\frac{N_{g}}{N}\overset{P}{\underset{N\rightarrow\infty}{\longrightarrow}%
}N_{g}^{\ast},\quad g=1,...,G.
\]
The \emph{semi-parametric estimator} of $\vartheta_{n^{\ast}}$, via QM$\phi
$Es, is%
\begin{equation}
\widetilde{\vartheta}_{\widehat{n}^{\ast},N,\phi}=\sum\limits_{g=1}^{G}%
w_{g}\widetilde{\vartheta}_{n_{g},N_{g},\phi}, \label{eq12}%
\end{equation}
where $w_{g}$\ is (\ref{eq4}),%
\[
\widehat{n}^{\ast}=\sum_{g=1}^{G}w_{g}n_{g},
\]%
\[
N=\sum_{g=1}^{G}N_{g},
\]%
\[
\widetilde{\vartheta}_{n_{g},N_{g},\phi}=\frac{X^{2}(\widetilde{\boldsymbol{Y}%
}_{g},\widehat{\boldsymbol{\theta}}_{\phi})}{(N_{g}-1)(M-1)},
\]%
\[
X^{2}(\widetilde{\boldsymbol{Y}}_{g},\widehat{\boldsymbol{\theta}}_{\phi
})=n_{g}\sum_{\ell=1}^{N_{g}}\sum_{r=1}^{M}\frac{(\widehat{p}_{r}^{(\ell
,g)}-\widehat{p}_{r}^{(g)})^{2}}{p_{r}(\widehat{\boldsymbol{\theta}}_{\phi}%
)}=n_{g}\sum_{r=1}^{M}\frac{1}{p_{r}(\widehat{\boldsymbol{\theta}}_{\phi}%
)}\sum_{\ell=1}^{N_{g}}(\widehat{p}_{r}^{(\ell,g)}-\widehat{p}_{r}^{(g)}%
)^{2}.
\]
Similarly, the semi-parametric estimator of $\rho^{2}$, via QM$\phi$Es, is%
\[
\widetilde{\rho}_{\widehat{n}^{\ast},N,\phi}^{2}=\frac{\widetilde{\vartheta
}_{\widehat{n}^{\ast},N,\phi}-1}{\widehat{n}^{\ast}-1}.
\]
Both, $\widetilde{\vartheta}_{\widehat{n}^{\ast},N,\phi}$ and $\widetilde{\rho
}_{\widehat{n}^{\ast},N,\phi}^{2}$, are consistent estimators of $\vartheta$
and $\rho^{2}$\ respectively.

The following results are not explicitly proven since the same steps of the
proof given in Section \ref{sec2} are needed. However, a basic and different
result is required in the place of (\ref{eq14}), which is%
\begin{equation}
\sqrt{N}\left(  \widehat{\boldsymbol{p}}-\boldsymbol{p}\left(
\boldsymbol{\theta}_{0}\right)  \right)  \overset{\mathcal{L}%
}{\underset{N\rightarrow\infty}{\longrightarrow}}\mathcal{N(}\boldsymbol{0}%
_{M},\tfrac{\vartheta_{n^{\ast}}}{\bar{n}}\boldsymbol{\Sigma}_{\boldsymbol{p}%
\left(  \boldsymbol{\theta}_{0}\right)  }), \label{eq16}%
\end{equation}
proven in Alonso-Revenga et al. (2016).

\begin{theorem}
\label{Th5}The asymptotic distribution of the difference between the
non-parametric estimator and QM$\phi$E of $\boldsymbol{p}(\boldsymbol{\theta
})$, with $G$ groups of clusters of size $n_{g}$, $g=1,...,G$, is%
\[
\sqrt{N}(\widehat{\boldsymbol{p}}-\boldsymbol{p}(\widehat{\boldsymbol{\theta}%
}_{\phi_{2}}))\overset{\mathcal{L}}{\underset{N\rightarrow\infty
}{\longrightarrow}}\mathcal{N}(\boldsymbol{0}_{M},\tfrac{\vartheta_{n^{\ast}}%
}{\bar{n}}(\boldsymbol{\Sigma}_{\boldsymbol{p}(\boldsymbol{\theta}_{0}%
)}-\boldsymbol{\Sigma}_{\boldsymbol{p}(\boldsymbol{\theta}_{0})}%
\boldsymbol{W}\left(  \boldsymbol{W}^{T}\boldsymbol{\Sigma}_{\boldsymbol{p}%
(\boldsymbol{\theta}_{0})}\boldsymbol{W}\right)  ^{-1}\boldsymbol{W}%
^{T}\boldsymbol{\Sigma}_{\boldsymbol{p}(\boldsymbol{\theta}_{0})})),
\]
where $\boldsymbol{\theta}_{0}$\ is the unknown true value of
$\boldsymbol{\theta}$, $\vartheta_{n^{\ast}}$ is (\ref{eq7}) and%
\[
\bar{n}=\sum_{g=1}^{G}N_{g}^{\ast}n_{g}.
\]

\end{theorem}

\begin{corollary}
\label{Cor6}The semiparametric clustered \emph{overdispersed chi-square GOF
test-statistic}, with $G$ groups of clusters of size $n_{g}$, $g=1,...,G$, has
the following asymptotic distribution%
\[
\frac{X^{2}(\boldsymbol{Y},\widehat{\boldsymbol{\theta}}_{\phi_{2}}%
)}{\widetilde{\vartheta}_{\widehat{n}^{\ast},N,\phi_{2}}}\overset{\mathcal{L}%
}{\underset{N\rightarrow\infty}{\longrightarrow}}\chi_{M-M_{0}-1}^{2},
\]
where $\widetilde{\vartheta}_{\widehat{n}^{\ast},N,\phi}$\ is (\ref{eq12}) and%
\[
X^{2}(\boldsymbol{Y},\widehat{\boldsymbol{\theta}}_{\phi_{2}})=\widehat{\bar
{n}}N(\widehat{\boldsymbol{p}}-\boldsymbol{p}(\widehat{\boldsymbol{\theta}%
}_{\phi_{2}}))^{T}\boldsymbol{D}_{\boldsymbol{p}(\widehat{\boldsymbol{\theta}%
}_{\phi_{2}})}^{-1}(\widehat{\boldsymbol{p}}-\boldsymbol{p}%
(\widehat{\boldsymbol{\theta}}_{\phi_{2}}))=\widehat{\bar{n}}N\sum_{r=1}%
^{M}\frac{(\widehat{p}_{r}-p_{r}(\widehat{\boldsymbol{\theta}}_{\phi_{2}%
}))^{2}}{p_{r}(\widehat{\boldsymbol{\theta}}_{\phi_{2}})},
\]
with%
\begin{equation}
\widehat{\bar{n}}N=\sum_{g=1}^{G}N_{g}n_{g}. \label{eq11}%
\end{equation}

\end{corollary}

\begin{theorem}
\label{Th7}The semiparametric clustered \emph{overdispersed divergence based
GOF test-statistic}, with $G$ groups of clusters of size $n_{g}$, $g=1,...,G$,
has the following asymptotic distribution%
\[
\frac{T^{\phi_{1}}(\boldsymbol{Y},\widehat{\boldsymbol{\theta}}_{\phi_{2}}%
)}{\widetilde{\vartheta}_{\widehat{n}^{\ast},N,\phi_{2}}}\overset{\mathcal{L}%
}{\underset{N\rightarrow\infty}{\longrightarrow}}\chi_{M-M_{0}-1}^{2},
\]
where
\[
T^{\phi_{1}}(\boldsymbol{Y},\widehat{\boldsymbol{\theta}}_{\phi_{2}}%
)=\frac{2\widehat{\bar{n}}N}{\phi_{1}^{\prime\prime}(1)}\mathrm{d}_{\phi_{1}%
}(\widehat{\boldsymbol{p}},\boldsymbol{p}(\widehat{\boldsymbol{\theta}}%
_{\phi_{2}})),
\]
$\widetilde{\vartheta}_{\widehat{n}^{\ast},N,\phi_{2}}$\ is (\ref{eq12}) and
$\widehat{\bar{n}}N$\ (\ref{eq11}).
\end{theorem}

\begin{corollary}
\label{Cor8}The semiparametric clustered \emph{overdispersed likelihood-ratio
GOF test-statistic}, with $G$ groups of clusters of size $n_{g}$, $g=1,...,G$,
has the following asymptotic distribution%
\[
\frac{G^{2}(\boldsymbol{Y},\widehat{\boldsymbol{\theta}}_{\phi_{2}}%
)}{\widetilde{\vartheta}_{\widehat{n}^{\ast},N,\phi_{2}}}\overset{\mathcal{L}%
}{\underset{N\rightarrow\infty}{\longrightarrow}}\chi_{M-M_{0}-1}^{2},
\]
where
\[
G^{2}(\boldsymbol{Y},\widehat{\boldsymbol{\theta}}_{\phi_{2}})=\widehat{\bar
{n}}N\sum_{r=1}^{M}\widehat{p}_{r}\log\frac{\widehat{p}_{r}}{p_{r}%
(\widehat{\boldsymbol{\theta}}_{\phi_{2}})},
\]
$\widetilde{\vartheta}_{\widehat{n}^{\ast},N,\phi_{2}}$\ is (\ref{eq12}) and
$\widehat{\bar{n}}N$\ (\ref{eq11}).
\end{corollary}

Brier (1980) proposed using $G^{2}(\boldsymbol{Y},\widehat{\boldsymbol{\theta
}})/\widehat{\vartheta}_{\widehat{n}^{\ast},N}$ and $X^{2}(\boldsymbol{Y}%
,\widehat{\boldsymbol{\theta}})/\widehat{\vartheta}_{\widehat{n}^{\ast},N}$,
with%
\begin{equation}
\widehat{\vartheta}_{\widehat{n}^{\ast},N}=\sum\limits_{g=1}^{G}%
w_{g}\widehat{\vartheta}_{n_{g},N_{g}}, \label{Br}%
\end{equation}
where%
\[
\widehat{\vartheta}_{n_{g},N_{g}}=\frac{n_{g}}{(N_{g}-1)(M-1)}\sum_{r=1}%
^{M}\frac{1}{\widehat{p}_{r}^{(g)}}\sum_{\ell=1}^{N_{g}}(\widehat{p}%
_{r}^{(\ell,g)}-\widehat{p}_{r}^{(g)})^{2},
\]
to be applied for the Dirichlet-multinomial distribution. In a similar way as
Theorem \ref{Th7}, it is proven that
\begin{equation}
\frac{T^{\phi_{1}}(\boldsymbol{Y},\widehat{\boldsymbol{\theta}}_{\phi_{2}}%
)}{\widehat{\vartheta}_{n_{g},N_{g}}}\overset{\mathcal{L}%
}{\underset{N\rightarrow\infty}{\longrightarrow}}\chi_{M-M_{0}-1}^{2}.
\label{NP}%
\end{equation}

\section{Numerical example\label{sec4}}

From all the households located in $N=20$ neighborhoods\ around Montevideo
(Minnesota, US), some households were randomly selected: from $N_{1}=18$
neighborhoods $n_{1}=5$ houses were selected and from $N_{2}=2$ neighborhoods
$n_{2}=3$ houses. The neighborhoods are grouped into class $g=1$ or $g=2$
depending on the selected number of houses (neighborhood or cluster size),
$n_{1}=5$\ and $n_{2}=3$\ respectively. For the $\ell$-th neighborhood
($\ell=1,...,N_{g}$) of the $g$-th cluster size, in the $s$-th selected home
($s=1,...,n_{g}$), the family was questioned on two study interests:
satisfaction with the housing in the neighborhood as a whole ($X_{1s}%
^{(g,\ell)}$), and satisfaction with their own home ($X_{2s}^{(g,\ell)}$). For
both questions the responses were classified as unsatisfied ($US$), satisfied
($S$) or very satisfied ($VS$). In the sequel, we shall identify the
aforementioned categories of the ordinal variables, $X_{11}^{(g,\ell)}$ and
$X_{12}^{(g,\ell)}$, with numbers $1$, $2$, and $3$: for example, $(US,S)$\ is
associated with $(X_{11}^{(g,\ell)}$,$X_{12}^{(g,\ell)})=(1,2)$.

Under the null hypothesis of (\ref{eq5}), a family's classification according
to level of personal satisfaction is independent from its classification by
level of community satisfaction. The corresponding log-linear model, $\log
p_{ij}(\boldsymbol{\theta})=u+\theta_{1(i)}+\theta_{2(j)}$, for $i=1,...,I=3$,
$j=1,...,J=3$, has as design matrix and the unknown parameter vector%
\[
\boldsymbol{W}=%
\begin{pmatrix}
1 & 1 & 1 & 0 & 0 & 0 & -1 & -1 & -1\\
0 & 0 & 0 & 1 & 1 & 1 & -1 & -1 & -1\\
1 & 0 & -1 & 1 & 0 & -1 & 1 & 0 & -1\\
0 & 1 & -1 & 0 & 1 & -1 & 0 & 1 & -1
\end{pmatrix}
^{T}\quad\text{and}\quad\boldsymbol{\theta}=(\theta_{1(1)},\theta
_{1(2)},\theta_{2(1)},\theta_{2(2)})^{T}\text{.}%
\]
The corresponding data, given in Table \ref{t0}, are disaggregated based on
the number of houses $(g)$ and neighborhood identifications $(\ell)$ in $20$
rows, having each $M=9$ cells in lexicographical order. The $G=2$ groups of
clusters have respectively $n_{1}=5$ and $n_{2}=3$ families.%

\begin{table}[hbpt]  \tabcolsep2.8pt \small\centering
$%
\begin{tabular}
[c]{rrccccccccc}\hline
$g$ & $\ell$ & $Y_{11}^{\left(  g,\ell\right)  }$ & $Y_{12}^{\left(
g,\ell\right)  }$ & $Y_{13}^{\left(  g,\ell\right)  }$ & $Y_{21}^{\left(
g,\ell\right)  }$ & $Y_{22}^{\left(  g,\ell\right)  }$ & $Y_{23}^{\left(
g,\ell\right)  }$ & $Y_{31}^{\left(  g,\ell\right)  }$ & $Y_{32}^{\left(
g,\ell\right)  }$ & $Y_{33}^{\left(  g,\ell\right)  }$\\\hline
$1$ & $1$ & $1$ & $0$ & $0$ & $2$ & $2$ & $0$ & $0$ & $0$ & $0$\\
$1$ & $2$ & $1$ & $0$ & $0$ & $2$ & $2$ & $0$ & $0$ & $0$ & $0$\\
$1$ & $3$ & $0$ & $2$ & $0$ & $0$ & $2$ & $0$ & $0$ & $1$ & $0$\\
$1$ & $4$ & $0$ & $1$ & $0$ & $2$ & $1$ & $0$ & $1$ & $0$ & $0$\\
$1$ & $5$ & $0$ & $0$ & $0$ & $0$ & $4$ & $0$ & $0$ & $1$ & $0$\\
$1$ & $6$ & $1$ & $0$ & $0$ & $3$ & $1$ & $0$ & $0$ & $0$ & $0$\\
$1$ & $7$ & $3$ & $0$ & $0$ & $0$ & $1$ & $0$ & $0$ & $1$ & $0$\\
$1$ & $8$ & $1$ & $0$ & $0$ & $1$ & $3$ & $0$ & $0$ & $0$ & $0$\\
$1$ & $9$ & $3$ & $0$ & $0$ & $0$ & $0$ & $0$ & $1$ & $0$ & $1$\\
$1$ & $10$ & $0$ & $1$ & $0$ & $0$ & $3$ & $1$ & $0$ & $0$ & $0$\\
$1$ & $11$ & $1$ & $1$ & $0$ & $0$ & $2$ & $0$ & $1$ & $0$ & $0$\\
$1$ & $12$ & $0$ & $1$ & $0$ & $4$ & $0$ & $0$ & $0$ & $0$ & $0$\\
$1$ & $13$ & $0$ & $0$ & $0$ & $4$ & $1$ & $0$ & $0$ & $0$ & $0$\\
$1$ & $14$ & $0$ & $0$ & $0$ & $1$ & $2$ & $0$ & $0$ & $0$ & $2$\\
$1$ & $15$ & $2$ & $0$ & $0$ & $2$ & $1$ & $0$ & $0$ & $0$ & $0$\\
$1$ & $16$ & $0$ & $0$ & $0$ & $1$ & $1$ & $1$ & $0$ & $2$ & $0$\\
$1$ & $17$ & $2$ & $0$ & $0$ & $2$ & $1$ & $0$ & $0$ & $0$ & $0$\\
$1$ & $18$ & $2$ & $0$ & $0$ & $2$ & $0$ & $0$ & $1$ & $0$ & $0$\\
$2$ & $1$ & $1$ & $0$ & $0$ & $1$ & $1$ & $0$ & $0$ & $0$ & $0$\\
$2$ & $2$ & $0$ & $0$ & $0$ & $1$ & $0$ & $1$ & $0$ & $0$ & $1$\\\hline
\end{tabular}
\ \ \ \ \ \ \ $%
\caption{Housing satisfaction in neighbourhoods of Montevideo (Brier, 1980).\label{t0}}%
\end{table}%

For estimation and testing, the power divergence measures are considered, by
restricting $\phi$ from the family of convex\ functions to the subfamily%
\[
\phi_{\lambda}(x)=\left\{
\begin{array}
[c]{ll}%
\frac{1}{\lambda(1+\lambda)}\left[  x^{\lambda+1}-x-\lambda(x-1)\right]  , &
\lambda\notin\{-1,0\}\\
\lim_{\upsilon\rightarrow\lambda}\frac{1}{\upsilon(1+\upsilon)}\left[
x^{\upsilon+1}-x-\upsilon(x-1)\right]  , & \lambda\in\{-1,0\}
\end{array}
\right.  ,
\]
where $\lambda\in%
\mathbb{R}
$ is a tuning parameter. The expression of (\ref{eq6}) becomes%
\[
d_{\phi_{\lambda}}(\widehat{\boldsymbol{p}},\boldsymbol{p}(\boldsymbol{\theta
}))=\left\{
\begin{array}
[c]{ll}%
\frac{1}{\lambda(\lambda+1)}\left(
{\displaystyle\sum\limits_{r=1}^{M}}
\frac{\widehat{p}_{r}^{\lambda+1}}{p_{r}^{\lambda}\left(  \boldsymbol{\theta
}\right)  }-1\right)  , & \lambda\notin\{-1,0\}\\
d_{Kullback}(\boldsymbol{p}\left(  \boldsymbol{\theta}\right)
,\widehat{\boldsymbol{p}}), & \lambda=-1\\
d_{Kullback}(\widehat{\boldsymbol{p}},\boldsymbol{p}\left(  \boldsymbol{\theta
}\right)  ), & \lambda=0
\end{array}
\right.  ,
\]
in such a way that for each $\lambda\in%
\mathbb{R}
$\ a different divergence measure is obtained. The quasi\ minimum
power-divergence estimator (QMPE) of $\boldsymbol{\theta}$, is given by
$\widehat{\boldsymbol{\theta}}_{\phi_{\lambda_{2}}}=\arg\min_{\theta\in\Theta
}d_{\phi_{\lambda_{2}}}(\widehat{\boldsymbol{p}},\boldsymbol{p}\left(
\boldsymbol{\theta}\right)  )$, and the semiparametric clustered overdispersed
power-divergence based GOF test-statistic, based on
$\widehat{\boldsymbol{\theta}}_{\phi_{\lambda_{2}}}$, by%
\begin{equation}
\widetilde{T}_{\lambda_{1},\lambda_{2}}=\frac{2\widehat{\bar{n}}%
N\mathrm{d}_{\phi_{\lambda_{1}}}(\widehat{\boldsymbol{p}},\boldsymbol{p}%
(\widehat{\boldsymbol{\theta}}_{\phi_{\lambda_{2}}}))}{\widetilde{\vartheta
}_{\widehat{n}^{\ast},N,\phi_{\lambda_{2}}}}=\frac{2\widehat{\bar{n}}%
N}{\widetilde{\vartheta}_{\widehat{n}^{\ast},N,\phi_{\lambda_{2}}}\lambda
_{1}(\lambda_{1}+1)}\left(
{\displaystyle\sum\limits_{r=1}^{M}}
\frac{\widehat{p}_{r}^{\lambda_{1}+1}}{p_{r}^{\lambda_{1}}%
(\widehat{\boldsymbol{\theta}}_{\phi_{\lambda_{2}}})}-1\right)  ,\text{ for
}\lambda_{1}\notin\{-1,0\}, \label{crt}%
\end{equation}
where $\widetilde{\vartheta}_{\widehat{n}^{\ast},N,\phi_{\lambda_{2}}}$\ is
(\ref{eq12}) and $\widehat{\bar{n}}N$\ (\ref{eq11}). The expression of the
semiparametric clustered overdispersed power-divergence based GOF
test-statistic for $\lambda_{1}=0$ ($\widetilde{T}_{0,\lambda_{2}}%
=G^{2}(\boldsymbol{Y},\widehat{\boldsymbol{\theta}}_{\phi_{\lambda_{2}}%
})/\widetilde{\vartheta}_{\widehat{n}^{\ast},N,\phi_{\lambda_{2}}}$) is in
Corollary \ref{Cor8} and for the case of $\lambda_{1}=-1$ is given by%
\[
\widetilde{T}_{-1,\lambda_{2}}=\frac{\widehat{\bar{n}}N}{\widetilde{\vartheta
}_{\widehat{n}^{\ast},N,\phi_{\lambda_{2}}}}\sum_{r=1}^{M}p_{r}%
(\widehat{\boldsymbol{\theta}}_{\phi_{2}})\log\frac{p_{r}%
(\widehat{\boldsymbol{\theta}}_{\phi_{2}})}{\widehat{p}_{r}}.
\]
Notice that the case of $\lambda_{2}=0$ for the QMPE of $\boldsymbol{\theta}$,
matches the QMLE of $\boldsymbol{\theta}$, $\widehat{\boldsymbol{\theta}}$, or
equivalently the QM$\phi$E\ of $\boldsymbol{\theta}$ with $\phi(x)=x\log
x-x+1$, and from the case of $\lambda_{1}=1$ arises the semiparametric
clustered overdispersed chi-square GOF test-statistic $\widetilde{T}%
_{1,\lambda_{2}}=X^{2}(\boldsymbol{Y},\widehat{\boldsymbol{\theta}}%
_{\phi_{\lambda_{2}}})/\widetilde{\vartheta}_{\widehat{n}^{\ast}%
,N,\phi_{\lambda_{2}}}$\ given in Corollary \ref{Cor6}. All these
test-statistics are completely new when no distributional assumption is made,
and homogeneous intracluster correlation assumption is considered cell by cell
in all the clusters.%

\begin{table}[hbpt]  \tabcolsep2.8pt \small\centering
\begin{tabular}
[c]{ccccccc}\hline
$\overset{}{\widetilde{T}_{\lambda_{1},\lambda_{2}}}$ &  &  &  & $\lambda_{2}$
&  & \\
($p$-value) &  & $-0.5$ & $0$ & $2/3$ & $1$ & $2$\\\hline
& $-0.5$ & $7.5621$ & $11.2413$ & $15.6963$ & $17.6234$ & $22.1483$\\
&  & $(0.1090)$ & $(0.0240)$ & $(0.0035)$ & $(0.0015)$ & $(0.0002)$\\
& $0$ & $7.7504$ & $9.7014$ & $12.2489$ & $13.4095$ & $16.2120$\\
&  & $(0.1012)$ & $(0.0458)$ & $(0.0156)$ & $(0.0094)$ & $(0.0027)$\\
$\lambda_{1}$ & $2/3$ & $10.4138$ & $10.3330$ & $11.3428$ & $11.9922$ &
$13.7789$\\
&  & $(0.0340)$ & $(0.0352)$ & $(0.0230)$ & $(0.0174)$ & $(0.0080)$\\
& $1$ & $13.0422$ & $11.2813$ & $11.4143$ & $11.8302$ & $13.2202$\\
&  & $(0.0111)$ & $(0.0236)$ & $(0.0223)$ & $(0.0187)$ & $(0.0102)$\\
& $2$ & $33.6045$ & $17.5637$ & $13.0587$ & $12.5518$ & $12.6781$\\
&  & $(<0.0001)$ & $(0.0015)$ & $(0.0110)$ & $(0.0137)$ & $(0.0130)$\\\hline
$\overset{}{\widetilde{\vartheta}_{\widehat{n}^{\ast},N,\phi_{\lambda_{2}}}}$
&  & $2.1815$ & $1.5869$ & $1.3314$ & $1.2707$ & $1.1813$\\\hline
\end{tabular}
\caption{Values for the clustered overdispersed GOF test-statistic, via semi-parametric estimates of the design effect, with corresponding $p$-values.\label{t1}}%
\end{table}%

The Brier's non-parametric estimator of $\vartheta_{n}$ can be also plugged on
the clustered overdispersed GOF test-statistic,
\[
\widehat{T}_{\lambda_{1},\lambda_{2}}=2\widehat{\bar{n}}N\mathrm{d}%
_{\phi_{\lambda_{1}}}(\widehat{\boldsymbol{p}},\boldsymbol{p}%
(\widehat{\boldsymbol{\theta}}_{\phi_{\lambda_{2}}}))/\widehat{\vartheta
}_{\widehat{n}^{\ast},N},
\]
with no change in the asymptotic distribution. In particular,
\[
\widehat{T}_{0,0}=G^{2}(\boldsymbol{Y},\widehat{\boldsymbol{\theta}}%
_{\phi_{\lambda_{2}}})/\widehat{\vartheta}_{\widehat{n}^{\ast},N}%
\qquad\text{and}\qquad\widehat{T}_{1,0}=X^{2}(\boldsymbol{Y}%
,\widehat{\boldsymbol{\theta}}_{\phi_{\lambda_{2}}})/\widehat{\vartheta
}_{\widehat{n}^{\ast},N}%
\]
are the clustered overdispersed GOF test-statistics proposed by Brier (1980).%

\begin{table}[hbpt]  \tabcolsep2.8pt \small\centering
\begin{tabular}
[c]{ccccccc}\hline
$\overset{}{\widehat{T}_{\lambda_{1},\lambda_{2}}}$ &  &  &  & $\lambda_{2}$ &
& \\
($p$-value) &  & $-0.5$ & $0$ & $2/3$ & $1$ & $2$\\\hline
& $-0.5$ & $15.4857$ & $16.7462$ & $19.6173$ & $21.0219$ & $24.5600$\\
&  & $(0.0038)$ & $(0.0022)$ & $(0.0006)$ & $(0.0003)$ & $(0.0001)$\\
& $0$ & $15.8714$ & $14.4521$ & $15.3087$ & $15.9953$ & $17.9773$\\
&  & $(0.0032)$ & $(0.0060)$ & $(0.0041)$ & $(0.0030)$ & $(0.0012)$\\
$\lambda_{1}$ & $2/3$ & $21.3256$ & $15.3931$ & $14.1762$ & $14.3048$ &
$15.2792$\\
&  & $(0.0003)$ & $(0.0040)$ & $(0.0068)$ & $(0.0064)$ & $(0.0042)$\\
& $1$ & $26.7079$ & $16.8057$ & $14.2656$ & $14.1115$ & $14.6597$\\
&  & $(<0.0001)$ & $(0.0021)$ & $(0.0065)$ & $(0.0069)$ & $(0.0055)$\\
& $2$ & $68.8157$ & $26.1646$ & $16.3207$ & $14.9723$ & $14.0586$\\
&  & $(<0.0001)$ & $(<0.0001)$ & $(0.0026)$ & $(0.0048)$ & $(0.0071)$\\\hline
$\overset{}{\widehat{\vartheta}_{\widehat{n}^{\ast},N}}$ &  & $1.0653$ &
$1.0653$ & $1.0653$ & $1.0653$ & $1.0653$\\\hline
\end{tabular}
\caption{Values for the clustered overdispersed GOF test-statistic, via non-parametric estimates of the design effect, with corresponding $p$-values.\label{t2}}%
\end{table}%

From the $p$-values of Tables \ref{t1} and \ref{t2} is concluded that only
$\widetilde{T}_{-0.5,-05}$ and $\widetilde{T}_{0,-0.5}$ clustered
overdispersed GOF test-statistics do not allow rejecting the null hypothesis.

\section{Simulation Study\label{sec5}}

In the simulation study performed In Section 6.1 of Alonso-Revenga et al.
(2016), a clear improvement of the semi-parametric estimator of $\rho^{2}$,
via QM$\phi$Es, $\widetilde{\rho}_{\widehat{n}^{\ast},N,\phi}^{2}$, was shown
in comparison with the Brier's non-parametric estimator of $\rho^{2}$,
$\widehat{\rho}_{\widehat{n}^{\ast},N}^{2}=(\widehat{\vartheta}_{\widehat{n}%
^{\ast},N}-1)/(\widehat{n}^{\ast}-1)$. Taking into account the same simulation
experiment, $\boldsymbol{\theta}=(\theta_{1(1)},\theta_{1(2)},\theta
_{2(1)},\theta_{2(2)})^{T}=(0.1,0.2,0.4,0.3)^{T}$ is the true value of the
parameter of the independence model described in Section \ref{sec4}, under the
null hypothesis. The study considers $G=3$ different cluster sizes with
$N_{1}=18$, $N_{2}=2$, $N_{3}=5$ clusters, having each $n_{1}=5$, $n_{2}=3$,
$n_{3}=7$ possibly correlated\ individuals.

With $R=10,000$\ replications the significance levels are estimated by
simulation for the power divergence based GOF test-statistics $\widetilde{T}%
_{\lambda_{1},\lambda_{2}}$ and $\widehat{T}_{\lambda_{1},\lambda_{2}}$, with
$\lambda_{1},\lambda_{2}\in\{-0.5,0,2/3,1,2\}$, defined in Section \ref{sec4}.
An extensive study has been done by considering three possible distributions
for $\boldsymbol{Y}^{(\ell)}$ but in Figure \ref{fig1} only a summary of the
final plots are shown. The three distributions, Dirichlet-multinomial (DM),
random-clumped (RC) and $n$-inflated (NI), mentioned in Section \ref{sec1},
are generated according to the algorithms described in Alonso-Revenga et al.
(2016) and Raim et al. (2015).

From the study it is concluded that a good behaviour of the estimator of
$\rho^{2}$ (or $\vartheta_{n}$) plays a crucial role on the behavoiur of the
closeness of the estimated significance level with respect to the nominal
significance level, but the choice of $\lambda_{1}=2/3$ for the GOF
test-statistic is also important. The combination of $\lambda_{1}=2/3$ for the
GOF test-statistic with $\lambda_{2}=2$ for the estimator in $\widetilde{T}%
_{\lambda_{1},\lambda_{2}}$ (or $\lambda_{1}=2/3$ for the GOF test-statistic
with $\lambda_{2}=0$ for the estimator) does not suffer negative modifications
as the value of $\rho^{2}$ increases in the abscissa axis. The Brier's
non-parametric estimator has however a negative impact on the estimated
significance levels of the classical overdispersed likelihood-ratio GOF
test-statistic $\widehat{T}_{0,0}=G^{2}(\boldsymbol{Y}%
,\widehat{\boldsymbol{\theta}}_{\phi_{\lambda_{2}}})/\widehat{\vartheta
}_{\widehat{n}^{\ast},N}$\ as the value of $\rho^{2}$ increases in the
abscissa axis. Looking at the right hand side plots, the three distributions
have estimated significance levels no closer to the nominal level, $0.05$, in
comparison with the rest of the distributions. In particular for $\lambda
_{1}=2/3$ and $\lambda_{2}=2$ with the $n$-inflated distribution the estimated
significance level tends to be below the nominal significance level, while for
the Dirichlet-multinomial and random-clumped distribution, above the nominal
significance level.%

\begin{figure}[htbp]  \centering
\begin{tabular}
[c]{c}%
${%
\begin{tabular}
[c]{c}%
{\includegraphics[
height=2.6567in,
width=3.5284in
]%
{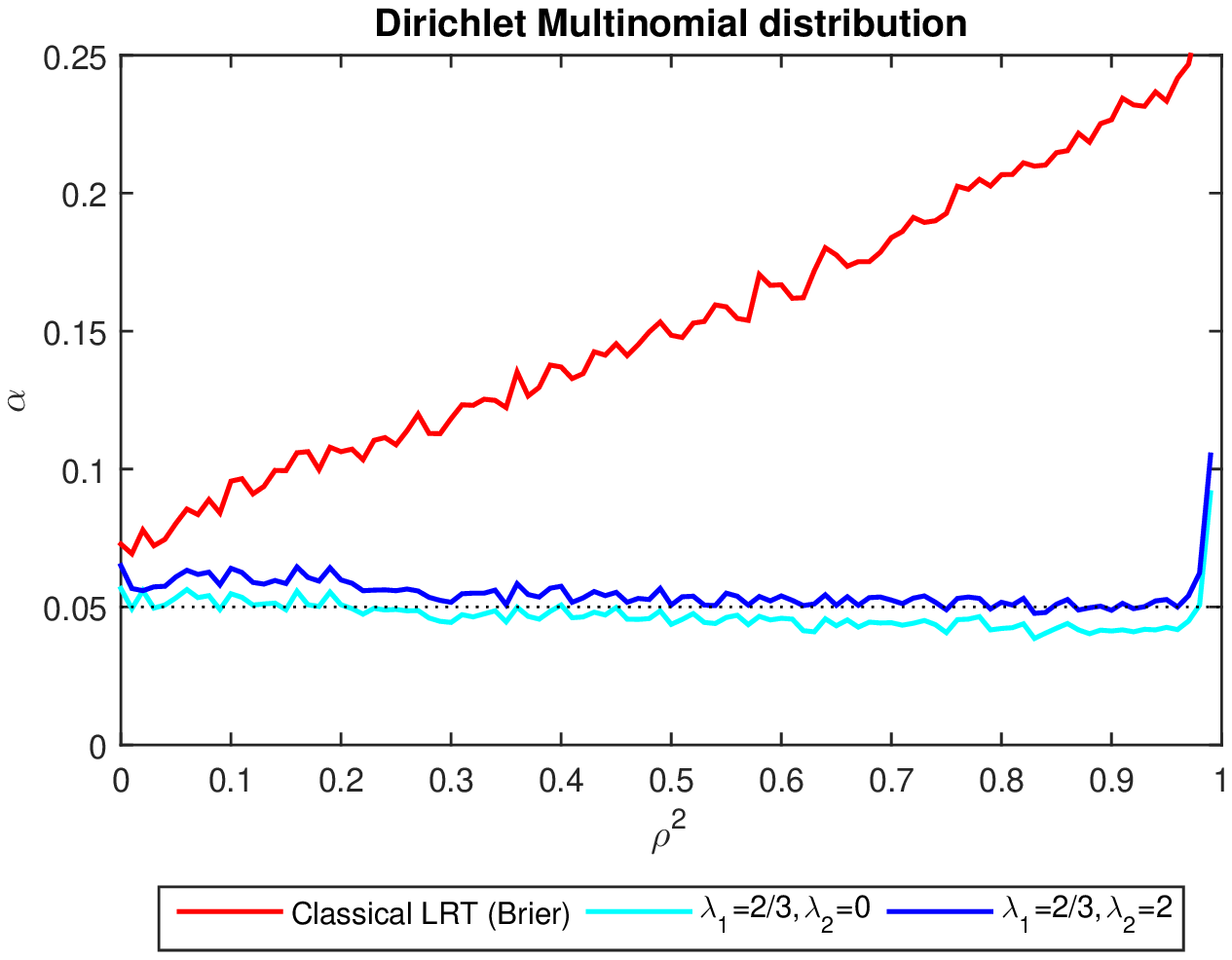}%
}
{\includegraphics[
height=2.6576in,
width=3.5284in
]%
{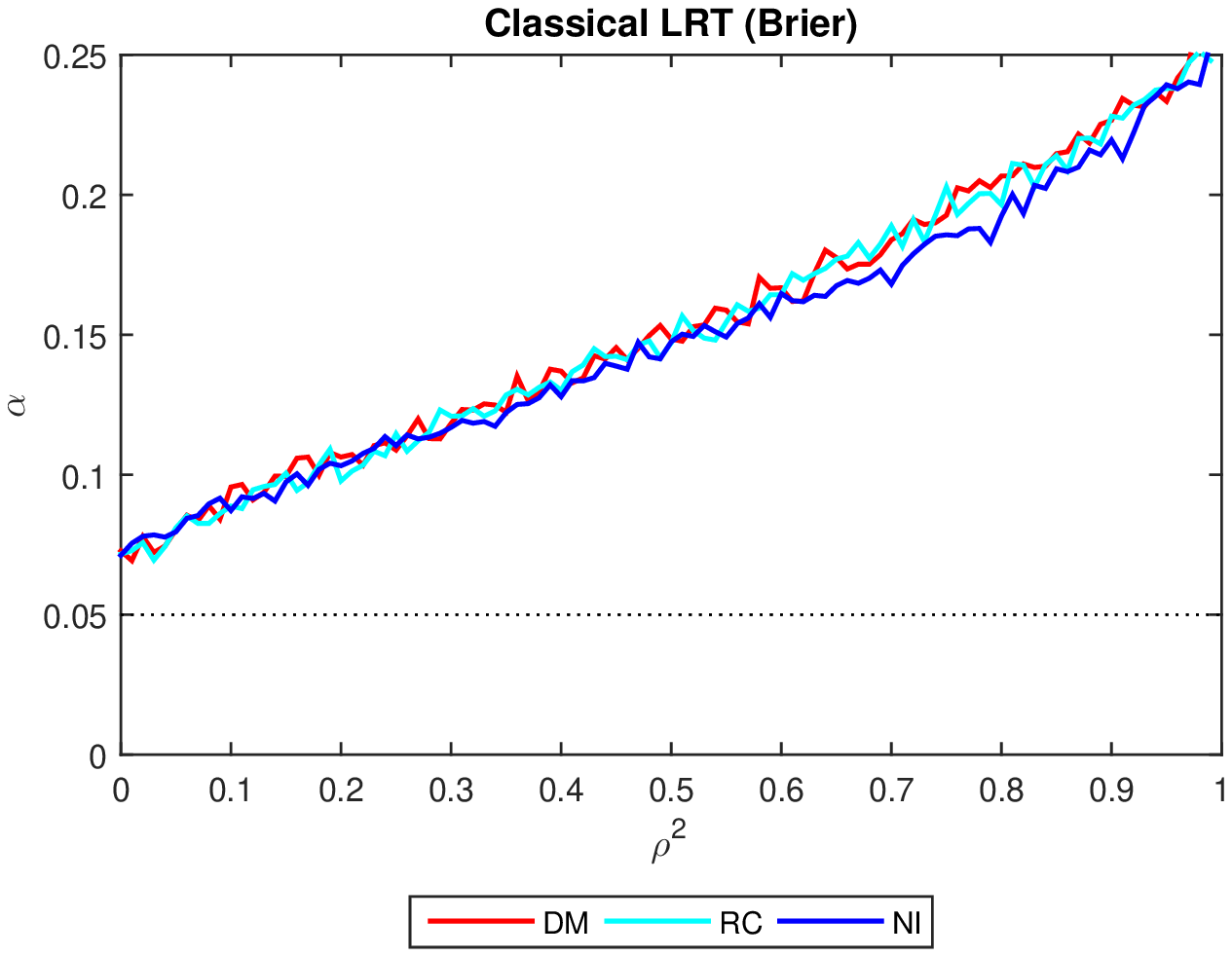}%
}
\\%
{\includegraphics[
height=2.6567in,
width=3.5284in
]%
{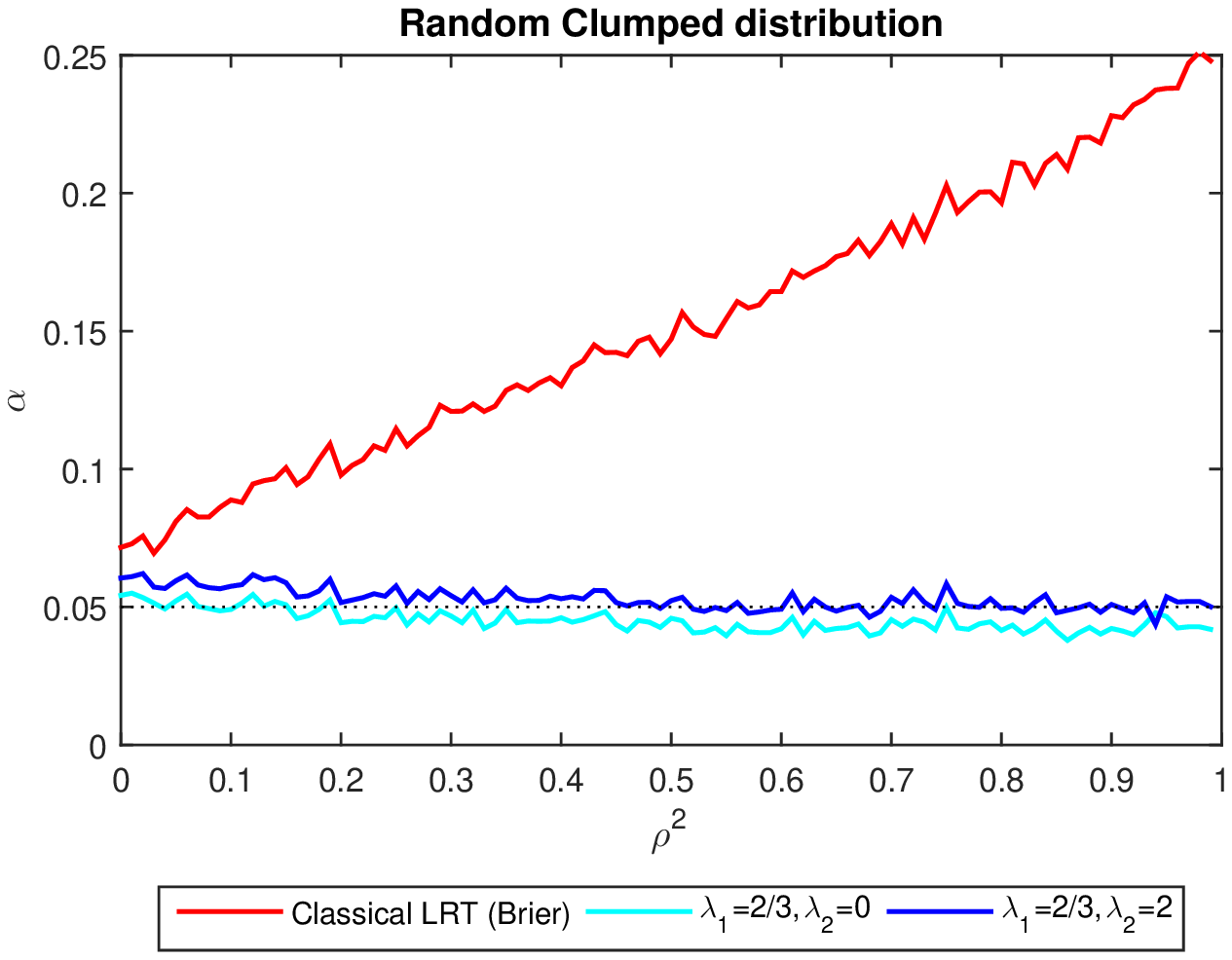}%
}
{\includegraphics[
height=2.6567in,
width=3.5284in
]%
{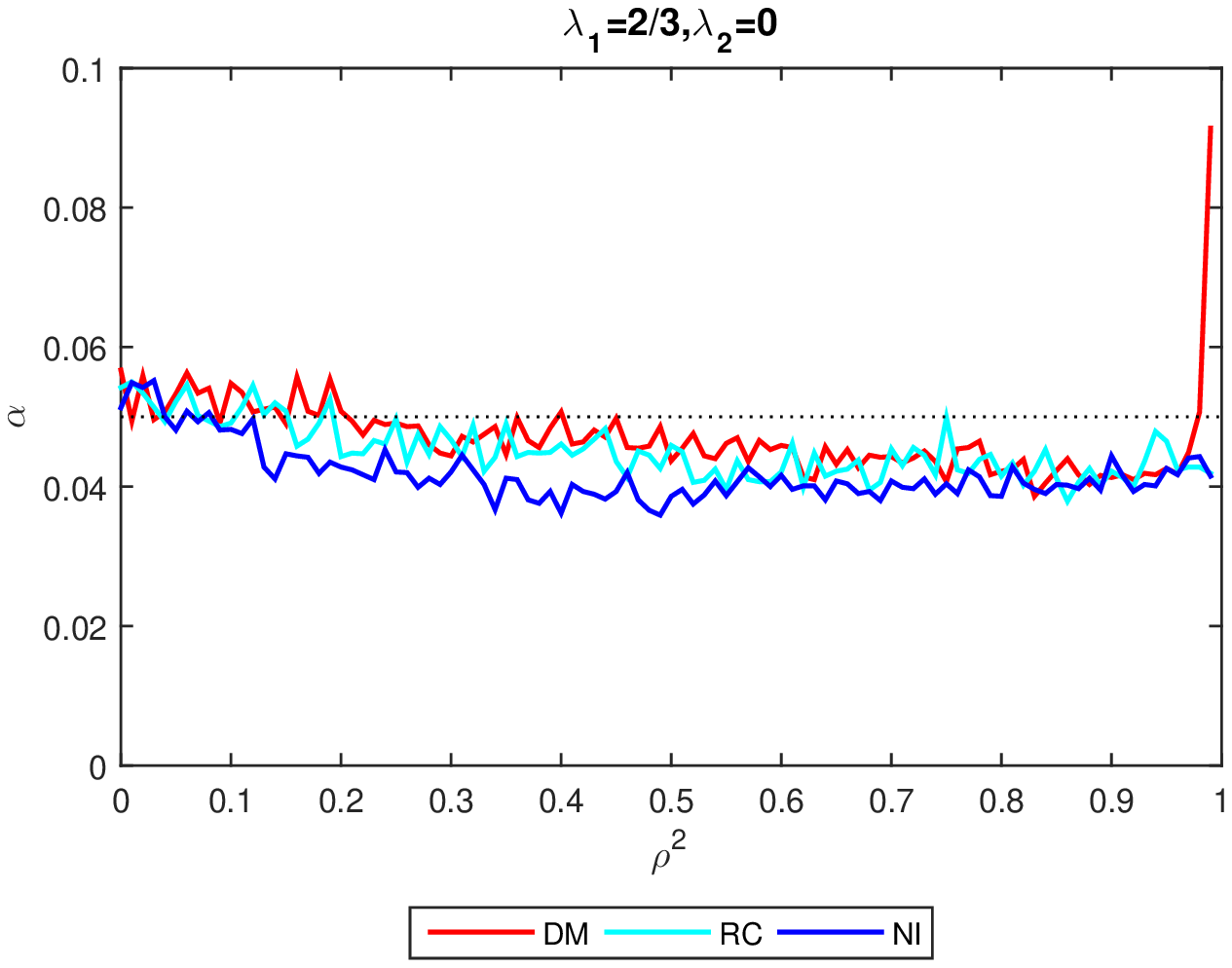}%
}
\\%
{\includegraphics[
height=2.6567in,
width=3.5284in
]%
{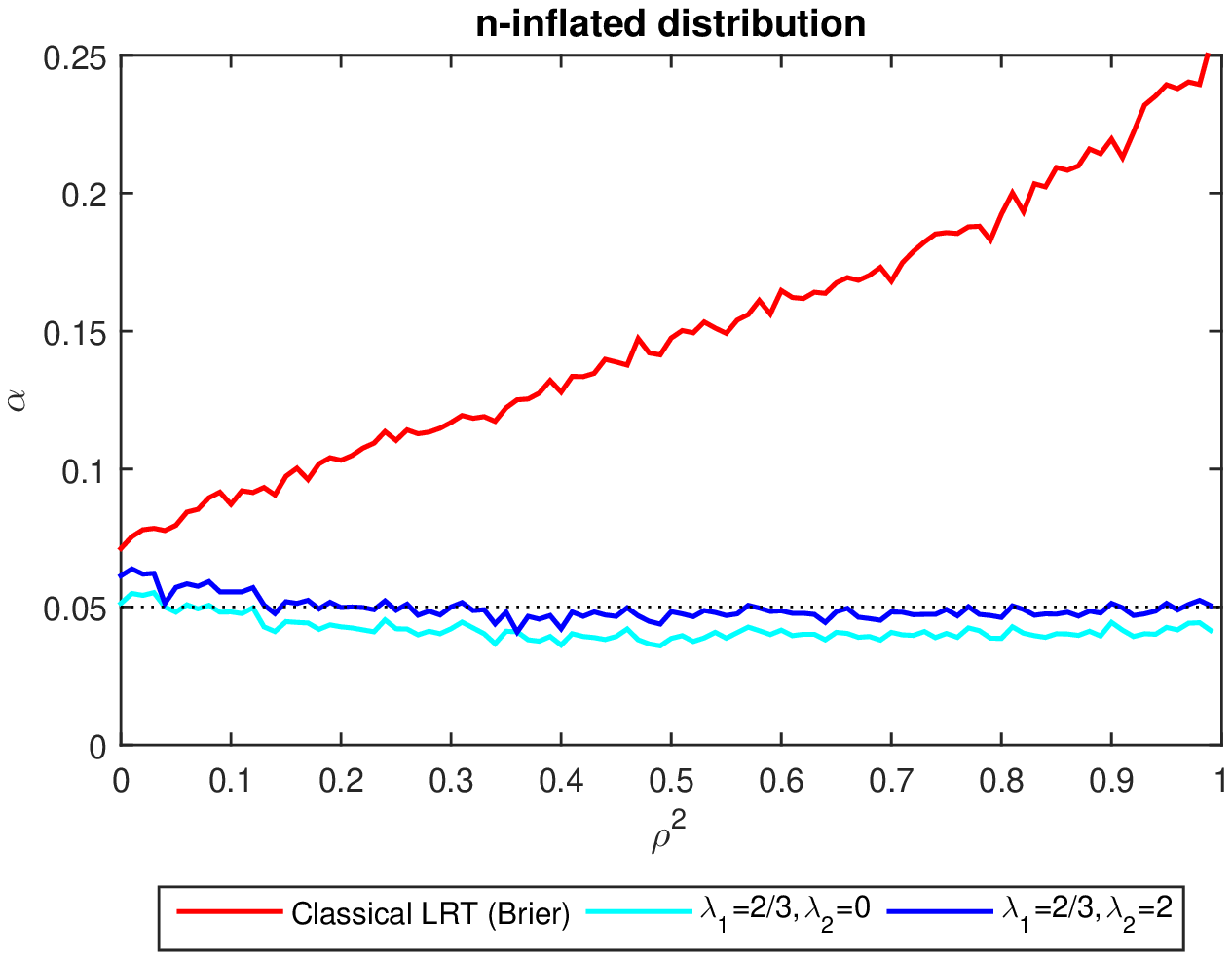}%
}
{\includegraphics[
height=2.6567in,
width=3.5284in
]%
{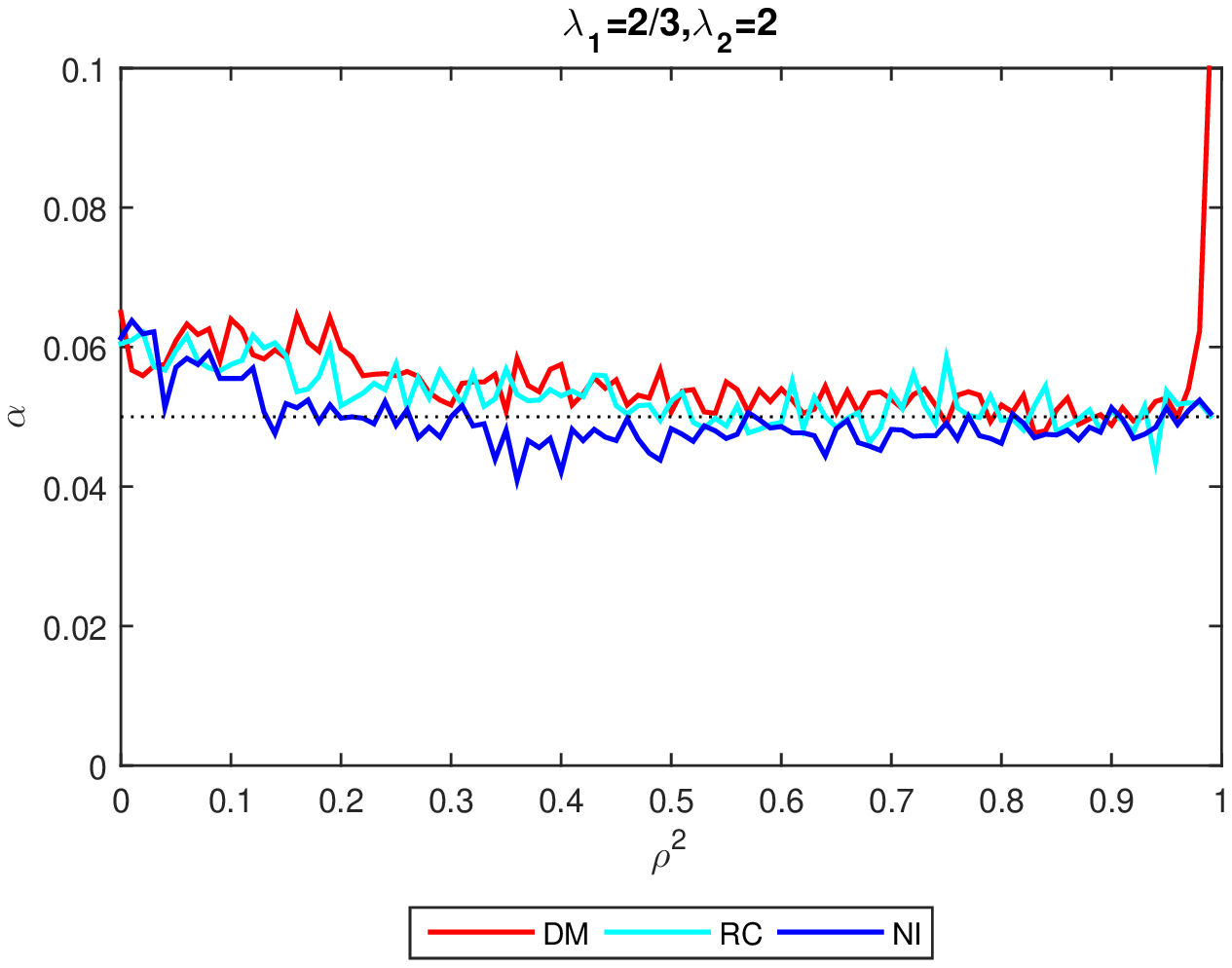}%
}
\end{tabular}
}$%
\end{tabular}
\caption{Estimated significance levels, by simulation, for three different distributions and types of overdispersed GOF test-statistics. \label{fig1}}%
\end{figure}%

\end{document}